\documentclass{birkjour}
\usepackage{bbold,relsize,mathabx}

\theoremstyle{theorem}
\newtheorem*{lemman}{Lemma}
\newtheorem{theom}{Theorem}
\def\Bbb#1{\mathbb{#1}}

\def\EINS{\mathbb{1}}

\begin{document}
\title[Resonance-Decay Problem in Quantum Mechanics]{Remarks to the Resonance-Decay Problem in Quantum Mechanics from a Mathematical Point of View}
\author{Hellmut Baumg\"artel}
\address{Mathematical Institute\\
University of Potsdam\\
Germany}
\email{baumg@uni-potsdam.de}

\begin{abstract}
The description of bumps in scattering cross-sections by Breit-Wigner amplitudes led in the framework of the mathematical Physics to its formulation as the so-called Resonance-Decay Problem. It consists of a spectral theoretical component and the connection of this component with the construction of decaying states. First the note quotes a solution for scattering systems, where the absolutely continuous parts of the Hamiltonians are semi-bounded and the scattering matrix is holomorphic in the upper half plane. This result uses the approach developed by Lax and Phillips, where the energy scale is extended to the whole real axis. The relationship of the spectral theoretic part of its solution and corresponding solutions obtained by other approaches is explained in the case of the Friedrichs model. A No-Go theorem shows the impossibility of the total solution within the specific framework of non-relativistic quantum mechanics. This points to the importance of the Lax-Phillips approach. At last, a solution is presented, where the scattering matrix is meromorphic in the upper half plane.
\end{abstract}

\subjclass{81U20 Scattering Theory}
\keywords{Resonance-Decay problem, Lax-Phillips approach, Friedrichs model}
\maketitle

\section{Introduction}

The origin of the resonance-decay problem in non-relativistic quantum mechanics is the observation of bumps in scattering cross-sections and their successful description by the so-called Breit-Wigner formula
\[
\Bbb{R}\ni\lambda\rightarrow\frac{1}{\pi}\frac{\alpha}{(\lambda-c)^{2}+\alpha^{2}},\quad\alpha>0,\;c\in\Bbb{R},
\]
where $c$ is the resonance energy $E_{0}$. Since
\begin{equation}\label{eq:1}
\frac{1}{\pi}\int_{-\infty}^{\infty}\frac{\alpha}{(\lambda-c)^{2}+\alpha^{2}}d\lambda=1,
\end{equation}
the Breit-Wigner formula describes, in a more general scope, a probability distribution of a real numerical quantity varying over the whole real axis $\Bbb{R}$. However in the framework of scattering theory the Breit-Wigner formula is interpreted by the Breit-Wigner amplitude
\[
e(\lambda):=(\frac{\alpha}{\pi})^{1/2}\frac{1}{\lambda-(c-i\alpha)},\quad c-i\alpha=:\zeta\in\Bbb{C}_{-},
\]
such that eq.~(1) now reads $\int_{-\infty}^{\infty}\vert e(\lambda)\vert^{2}d\lambda=1$. This suggests to consider the function $e$ as a state of a quantum-mechanical quantity, whose states are elements of the Hilbert space $L^{2}(\Bbb{R},d\lambda)$, where the energy is ``diagonalized'', i.e. the energy operator is the multiplication operator $M$ in this space and the time-evolution is given by the unitary operator $e^{-itM}$.

The so-called expectation value for the state $e$ according to the evolution $e^{-itM}$ is given by $(e,e^{-itM}e)$ and the corresponding ``Born probability'' by $\vert(e,e^{-itM}e)\vert^{2}$. The calculation of the expectation value gives
\[
(e,e^{-itM}e)=e^{-\alpha\vert t\vert-ict},
\]
i.e. with $\zeta:=c-i\alpha$ one gets
\begin{equation}\label{eq:2}
(e,e^{-itM}e)=
\begin{cases}
e^{-it\zeta},\quad t>0,\\
e^{-it\overline{\zeta}},\quad t<0,
\end{cases}
\end{equation}
and
\[
\vert(e,e^{-itM}e)\vert^{2}=e^{-2\alpha\vert t\vert},\quad t\in\Bbb{R}.
\]
That is, the expectation value forms an exponentially decaying semigroup $t\rightarrow e^{-it\zeta}$ for $t\geq 0$, similarly for $t\leq 0$. Within this context the Breit-Wigner amplitude - in particular the decay-semigroup property~\eqref{eq:2} - suggests the idea that $e$ could be interpreted as an unstable or decaying state, i.e. as an eigenvector of an exponentially decaying semigroup, where $\zeta$ is an eigenvalue of its generator. This interpretation of the Breit-Wigner amplitude leads to the problem to derive such a semigroup and the spectrum of its generator from properties of the scattering matrix, in particular from their poles in the lower half plane. The reason to focus on the poles is explained in Sec.2. The formulation of these ideas within the framework of the mathematical scattering theory led to the so-called

\paragraph{Resonance-Decay Problem} Let $\{H,H_{0}\}$ be an asymptotically complete quan\-tum-mechanical scattering system with scattering operator $S$. Then one has to construct a non-selfadjoint operator $B$, generator of a so-called decay-semigroup, depending on $H$ via $S$, whose eigenvalue spectrum coincides with the set of all poles of the scattering matrix in the lower half plane, such that the corresponding eigenstates can be interpreted as the hypothetical decaying states, connected with the Breit-Wigner amplitudes.

In other words, the first part of the problem is a {\em spectral theoretical characterization} of the poles of the scattering matrix in the lower half plane and the second step is to connect this characterization with the decay problem. Since the seventies this problem induced a vast series of various developments (for a selection of references see~\cite{ref:1,ref:2}). In the non-relativistic quantum mechanics the focus of interest is directed to scattering systems, whose Hamiltonians are semi-bounded. However the mathematical framework of scattering theory also includes scattering systems, where the generator $H$ of the unitary evolution group $\Bbb{R}\ni t\rightarrow 
e^{-itH}$ on a Hilbert space $\mathcal{H}$ together with the multiplication operator $M=:H_{0}$ on the Hilbert space 
$L^{2}(\Bbb{R},\mathcal{K},d\lambda)$ form an asymptotic complete scattering system, i.e. the absolutely continuous spectrum of $H$ and $H_{0}$ is the whole real axis. The Hilbert space $\mathcal{K}$ represents the multiplicity of the problem.

A scattering theory for such unitary evolutions, which are equipped with so-called outgoing and incoming subspaces was presented by P.D. Lax and R.S. Phillips (see~\cite{ref:3} and also~\cite[Chap. 12]{ref:6}), in particular for mutual orthogonal out- and incoming subspaces. In this case the scattering matrix is holomorphic in $\Bbb{C}_{+}$, it satisfies the condition $\Vert S(z)\Vert\leq 1,z\in\Bbb{C}_{+}$, and Lax and Phillips solved the resonance-decay problem completely. A decisive concept for their solution is an invariant subspace of a special
decay-semigroup, defined on the so-called Hardy-space $\mathcal{H}_{+}^{2}(\Bbb{R},\mathcal{K},d\lambda)$ for the upper half plane, defined by
\begin{equation}\label{eq:3}
\Bbb{R}\ni t\rightarrow Q_{+}e^{-itM}\mathlarger{\mathlarger{\restriction}}_{\mathcal{H}_{+}^{2}},
\end{equation}
where $Q_{+}$ is the projection from $L^{2}(\Bbb{R},\mathcal{K},d\lambda)$ onto $\mathcal{H}^{2}_{+}$. In the following this decay-semigroup is called the characteristic semigroup. At this point it is appropriate to refer to the property~\eqref{eq:2} of the Breit-Wigner amplitude and the idea mentioned there. Remarkably $e$ is an eigenvector of the decay-semigroup~\eqref{eq:3} with eigenvalue $e^{-i\zeta t}$. I.e. the Breit-Wigner amplitude appears in the Lax-Phillips approach as an important element of a solution of the resonance-decay problem.

The idea, to use the Lax-Phillips technique also for the resonance-decay problem on the positive energy half-axis, i.e. for semi-bounded Hamiltonians, is based on the property that the linear manifold $P_{+}\mathcal{H}^{2}_{+}$ is dense in the Hilbert space $L^{2}(\Bbb{R}_{+},\mathcal{K},d\lambda)$ of the reference Hamiltonian $H_{0}$, where $P_{+}$ is the projection from $L^{2}(\Bbb{R},\mathcal{K},d\lambda)$ on $L^{2}(\Bbb{R}_{+},\mathcal{K},d\lambda)$. The linear manifold $P_{+}\mathcal{H}^{2}_{+}$ equipped with the norm of $\mathcal{H}^{2}_{+}$ {\em is} $\mathcal{H}^{2}_{+}$ itself, since $P_{+}$ is injective on $\mathcal{H}^{2}_{+}$. The proof of the following result for semi-bounded Hamiltonians $H,H_{0}$ applies the Lax-Phillips technique, in particular the properties of the characteristic semigroup (see~\cite[Theorem 3]{ref:2}):

If the scattering matrix of the scattering system $\{H,H_{0}\}$ on $\Bbb{R}_{+}$ is holomorphic continuable into the upper half plane $\Bbb{C}_{+}$ and satisfies a certain boundedness condition then an invariant subspace of the characteristic semigroup is constructed such that its restriction to this subspace is a solution of the resonance-decay problem. That is, the spectrum of the generator of this restriction consists exactly of all poles of the scattering matrix and its resolvent set of all points, where the scattering matrix is holomorphic. The eigenvectors for the poles $\zeta$ are exactly special Breit-Wigner amplitudes $k/(\lambda-\zeta)$ for certain vectors $k$, which satisfy the condition $S(\overline{\zeta})^{\ast}k=0$, thus solving the multiplicity problem.

So far, this result solves the first part of the resonance-decay problem, the spectral characterization of the poles, completely. However, the eigenvectors, i.e. the decaying states, are vectors from $\mathcal{H}^{2}_{+}\subset L^{2}(\Bbb{R},\mathcal{K},d\lambda)$. This means that the direct relationship to the positive energy axis and the initial Hilbert space of states is lost. Even it turns out that these states constructed, i.e. the Breit-Wigner amplitudes, cannot be transferred unitarily to the Hilbert space $L^{2}(\Bbb{R}_{+},\mathcal{K},d\lambda)$ of the reference Hamiltonian $H_{0}=M_{+}$, its multiplication operator, to be decaying states w.r.t. the evolution $e^{-itM_{+}}$ (see Sec.~3).

Nevertheless, in sight of the structural mathematical point of view, the result can be interpreted within this Hilbert space: in the dense linear manifold $P_{+}\mathcal{H}^{2}_{+}$ one can introduce without ambiguity the $\mathcal{H}^{2}_{+}$-norm. Then the decay-semigroup can be considered as acting on $P_{+}\mathcal{H}^{2}_{+}$ equipped with this stronger norm.

The real weakness of this result is that its ansatz is an abstract one. Basically, it takes into consideration only the (canonical) reference Hamiltonian $H_{0}$ for the absolutely continuous spectrum and the scattering operator. The right or justification to use this ansatz goes back to the theorem of Wollenberg (see~\cite{ref:5,ref:6}) together with the fact that the Hamiltonian $H$, i.e. the interaction, is sometimes unknown. However, the reason for the appearance of poles in the scattering matrix of a scattering system ${H,H_{0}}$ remains to be seen.

\section{Friedrichs Model}

In so-called Friedrichs models the reason for the appearance of poles in the corresponding scattering matrix can be recognized. In these models eigenvalues of the reference operator $H_{0}$, which are embedded in its absolutely continuous spectrum are sometimes unstable caused by the interaction of $H$. They can generate poles of the scattering matrix. Therefore, the spectral theoretical characterization of these poles can be alternatively obtained by the method of ``generalized eigenvalues'', for example by so-called Gelfand triples. In this approach the corresponding eigenantilinear forms are usually of the pure Dirac-type (see e.g.~\cite{ref:4,ref:7,ref:9} and~\cite{ref:8}). Interestingly for the Friedrichs model presented in~\cite{ref:7} the solution of the multiplicity problem corresponds exactly to the solution quoted in Sec.~1 (see \cite[Theorem 4.2]{ref:7}). In this paper the Friedrichs model $H:=M+\Gamma +\Gamma^{\ast}$ on the full energy axis is considered on the Hilbert state space $\mathcal{H}:=L^{2}(\Bbb{R},\mathcal{K},d\lambda)\oplus\mathcal{E}$, where $\dim\mathcal{K}<\infty,\dim\mathcal{E}<\infty$ and $M$ the multiplication operator. The operator $\Gamma:\mathcal{E}\rightarrow L^{2}(\Bbb{R},\mathcal{K},d\lambda)$ is defined by $(\Gamma e)(\lambda):=M(\lambda)e$, where $M(\lambda)\in\mathcal{L}(\mathcal{E}\rightarrow\mathcal{K})$. In this model one obtains for the scattering matrix the expression
\begin{equation}\label{eq:4}
S_\mathcal{K}(\lambda)=\EINS_\mathcal{K}-2\pi iM(\lambda)L_{+}(\lambda+i0)^{-1}M(\lambda)^{\ast},
\end{equation}
where $L_{+}(\cdot)$ denotes the so-called Liv\v{s}ic-matrix on $\Bbb{C}_{+}$. Now, if one assumes that $M(\cdot)$ is holomorphic on $\Bbb{R}$ and meromorphic continuable, then for a pole ${\zeta}$ of the scattering matrix in $\Bbb{C}_{-}$ one obtains that the multiplicity of the corresponding Dirac-antilinear form is given by those $k\in\mathcal{K}$, such that $k=M(\zeta)e, e\in\mathcal{E}$, where $L_{+}(\zeta)e=0$ and $L_{+}(\cdot)$ denotes the continuation of the Liv\v{s}ic-matrix into the lower half plane (see~\cite[Theorem 4.2]{ref:7}). The correspondence to the solution of the multiplicity problem quoted in Sec.~1 is expressed by
\begin{lemman}
Let $\zeta\in\Bbb{C}_{-}$ and $k\in\mathcal{K}$. Then $S(\overline{\zeta})^{\ast}k=0$ iff $k=M(\zeta)e$, where $e\in\mathcal{E}$ {\em and} $L_{+}(\zeta)e=0$.
\end{lemman}
\begin{proof}
(i) Let $S(\overline{\zeta})^{\ast}k=0$. Then, according to equ.~\eqref{eq:4}, one obtains 
\[
k=-2\pi i M(\zeta)(L_{+}(\overline{\zeta})^{-1})^{\ast}M(\overline{\zeta})^{\ast}k
\]
Put
\begin{equation}\label{eq:5}
e:=-2\pi i(L_{+}(\overline{\zeta})^{-1})^{\ast}M(\overline{\zeta})^{\ast}k=
-2\pi i(L_{+}(\overline{\zeta})^{\ast})^{-1}M(\overline{\zeta})^{\ast}k.
\end{equation}
Since $\zeta\in\Bbb{C}_{-}$, the Liv\v{s}ic-matrix at $\overline{\zeta}$ reads
\[
L_{+}(\overline{\zeta})=
(\overline{\zeta}-H_{0})P_\mathcal{E}-\int_{-\infty}^{\infty}\frac{M(\lambda)^{\ast}M(\lambda)}{\overline{\zeta}-\lambda}d\lambda.
\]
Further note
\[ 
L_{+}(\overline{\zeta})^{\ast}=L_{-}(\zeta),
\]
where $L_{-}(\cdot)$ denotes the Liv\v{s}ic-matrix on $\Bbb{C}_{-}$.
According to equ.5 one has
\[
L_{+}(\overline{\zeta})^{\ast}e=-2\pi iM(\overline{\zeta})^{\ast}k=L_{-}(\zeta)e.
\]
For the continuation of $L_{+}(\cdot)$ into the lower half plane one obtains
\[
L_{+}(\zeta)=L_{-}(\zeta)+2\pi iM(\overline{\zeta})^{\ast}M(\zeta).
\]
Then one obtains
\[
L_{+}(\zeta)e=-2\pi iM(\overline{\zeta})^{\ast}(k-M(\zeta)e).
\]
The definition of $e$ in equ.5 implies $k=M(\zeta)e$. Therefore $L_{+}(\zeta)e=0$ follows.

\vspace{1mm}

(ii) Let $k:=M(\zeta)e$, where $L_{+}(\zeta)e=0$. Then
\[
L_{+}(\overline{\zeta})^{\ast}e=L_{-}(\zeta)e=-2\pi i M(\overline{\zeta})^{\ast}M(\zeta)e,
\]
hence
\[
e=-2\pi i(L_{+}(\overline{\zeta})^{\ast})^{-1}M(\overline{\zeta})^{\ast}M(\zeta)e
\]
and
\[
k=M(\zeta)e=-2\pi i M(\zeta)(L_{+}(\overline{\zeta})^{\ast})^{-1}M(\overline{\zeta})^{\ast}k
\]
follows, i.e. $S(\overline{\zeta})^{\ast}k=0.$ 
\end{proof}

Similar results one obtains for Friedrich models on the positive half axis.

The scattering matrices of Friedrichs models may have poles in the upper half plane. Insofar the extension of the result mentioned in Sec.~1 to these cases is obvious. For example, if in the Friedrichs model considered one puts $\dim\mathcal{K}=1$ and
$\Gamma e(\lambda):=\pi^{-1/2}(\lambda+i)^{-1}$, then $\zeta:=i$ is a pole of the scattering matrix.

\section{A No-Go-Theorem}

In Sec.~1 a solution of the resonance-decay problem for a scattering system on the positive half line, proved in~\cite{ref:2}, is quoted and critically considered. It was mentioned that the decaying states constructed there cannot be transferred to the Hilbert space of the reference Hamiltonian. This section contains a proof for this assertion.

\begin{theom}\label{theom:1}
There is no state $\phi\in L^{2}(\Bbb{R}_{+},\mathcal{K},d\lambda),\;\Vert\phi\Vert=1$, such that the Born-probability w.r.t. the unitary time evolution generated by $M_{+}$ is exponentially decaying, i.e. such that
\[
\vert(\phi,e^{-itM_{+}}\phi)\vert^{2}=e^{-2\alpha t},\quad t>0,
\]
for some constant $\alpha>0$.
\end{theom}
\begin{proof}
Born probabilities are symmetric w.r.t. future and past, i.e. they depend only on $\vert t\vert$. Assume that there is a state $\phi$ and a constant $\alpha>0$ such that
\[
\vert(\phi,e^{-itM_{+}}\phi)\vert^{2}=e^{-2\alpha\vert t\vert},\quad t\in\Bbb{R}.
\]
Then $\vert(\phi,e^{-itM_{+}}\phi)\vert=e^{-\alpha\vert t\vert}$ and
\[
(\phi,e^{-itM_{+}}\phi)=\int^{\infty}_{0}e^{-it\lambda}\vert\phi(\lambda)\vert^{2}_\mathcal{K}d\lambda=e^{-\alpha\vert t\vert+i\beta(t)},
\]
where $\beta(\cdot)$ is real-valued, continuous and one has $\beta(-t)=-\beta(t)$. Define
\[
g(\lambda)=\left\{
\begin{array}{ll}
\vert\phi(\lambda)\vert^{2}_\mathcal{K},\quad \lambda>0\\
0,\quad \lambda<0\\
\end{array}
\right.
\]
Then $g\in L^{1}(\Bbb{R},d\lambda)$ and
\[
\int^{\infty}_{-\infty}e^{-it\lambda}g(\lambda)d\lambda=e^{-\alpha\vert t\vert+i\beta(t)}.
\]
The function on the right hand side is a $L^{2}$-function, where
\[
\int_{-\infty}^{\infty}\alpha\vert e^{-\alpha\vert t\vert+i\beta(t)}\vert^{2}dt=1,
\]
i.e. there is a function $f\in L^{2}(\Bbb{R},d\lambda)$ such that
\[
F(f)(t)=\hat{f}(t)=\alpha^{1/2}e^{-\alpha\vert t\vert+i\beta(t)},
\]
where $\Vert f\Vert_{L^{2}}=1$ and $F$ denotes the Fourier transform. That is, one obtains
\[
\int_{-\infty}^{\infty}e^{-it\lambda}(2\pi\alpha)^{-1/2}f(\lambda)d\lambda=e^{-\alpha\vert t\vert+i\beta(t)}=
\int_{-\infty}^{\infty}e^{-it\lambda}g(\lambda)d\lambda.
\]
Since the Schwartz space $\mathcal{S}(\Bbb{R})$ is dense in $L^{1}(\Bbb{R})$ w.r.t. the $L^{1}$-norm and dense in $L^{2}(\Bbb{R})$ w.r.t. the $L^{2}$-norm, according to a standard argument in the theory of the Fourier transformation it follows that
$(2\pi\alpha)^{-1/2}f=g$, i.e. one obtains $g\in L^{2}(\Bbb{R})$.

Now $g$ has the property $g=P_{+}g$. This means $\hat{f}$ is an element of $\mathcal{H}^{2}_{-}(\Bbb{R})$, the Hardy space of the lower half plane, i.e. one gets $\hat{f}=Q_{-}\hat{f}$, where $Q_{-}$ denotes the projection onto $\mathcal{H}^{2}_{-}$. Because of
\[
P_{+}=F^{-1}Q_{-}F
\]
it follows that the inverse Fourier transform $F^{-1}\hat{f}$ is necessarily from\linebreak $P_{+}L^{2}(\Bbb{R},dx)$, i.e. the function
\[
x\rightarrow \int_{-\infty}^{\infty}e^{ixt}e^{-\alpha\vert t\vert+i\beta(t)}dt
\]
vanishes for $x<0$. However, the function
\[
h(z):=\int_{-\infty}^{\infty}e^{izt}e^{-\alpha\vert t\vert+i\beta(t)}dt
\]
is well-defined within the stripe $\vert\mbox{Im} z\vert<\alpha$ and a holomorphic function there. Therefore, since this function vanishes for $z=x<0$ it vanishes identically, hence also for $x>0$, i.e. one obtains $g=f=0$, a contradiction.
\end{proof}

\section{A Result for Scattering Systems with Poles of the Scattering Matrix in the Upper Half-plane}

An extension of the result mentioned in Sec.1 in this direction is suggested in Sec.~2. An incomplete version of the following result can be found already in~\cite[Theorem 2]{ref:2}, incomplete because of a flaw in the proof. Surprisingly it turns out that not only the poles in the lower half plane cause decaying states, but also holomorphic points $\zeta$ there may generate such states, but only in the case that
$\overline{\zeta}$ is a pole (in the upper half plane) with a special property of its main part.
\begin{theom}\label{theom:2}
Assume that the scattering matrix of the scattering system\linebreak $\{H,H_{0}\}$ on $\Bbb{R}$ satisfies the following
conditions:
\begin{enumerate}
\item[(I)] It is meromorphic in $\Bbb{C}_{+}$ with at most finitely many poles,
\item[(II)] $\Vert S(z)\Vert<K,\;K>0,\;z\in\Bbb{C}_{+}\;\vert z\vert>R$, where $R$ is sufficiently large,
\item[(III)] there are no complex-conjugated poles,
\item[(IV)] there is at least one pole in $\Bbb{C}_{-}$.
\end{enumerate}
Then the spectrum $\operatorname{spec}B_{+}\subset\Bbb{C}_{-}$ of the generator $B_{+}$ of the restriction of the characteristic semigroup to the subspace $\Bbb{T}_{+}\subset\mathcal{H}^{2}_{+}$ is described as follows:
\begin{enumerate}
\item[(i)] $\zeta\in\Bbb{C}_{-}$ is an eigenvalue of $B_{+}$ iff (a) $\zeta$ is a pole of $S(\cdot)$ {\em or} (b) $\overline{\zeta}$ is a pole of $S(\cdot)$ and the operator coefficient $A$ of the leading term of the main part of the pole $\overline{\zeta}$ is not invertible.

\item[(ii)] $\zeta\in\Bbb{C}_{-}$ is a point of the resolvent set res$\,B_{+}$ of $B_{+}$ iff (a) $S(\zeta)$ and $S(\overline{\zeta})$ exist, i.e. $\zeta,\overline{\zeta}$ are holomorphic points of $S(\cdot)$ or (b) $S(\zeta)$ 
exists and $\overline{\zeta}$ is a pole of $S(\cdot)$ and $A$ is invertible, i.e. $A^{-1}$ exists.
\end{enumerate}
\end{theom}
\begin{proof}
Let $\eta_{1},\eta_{2},...,\eta_{r}$ be the poles in $\Bbb{C}_{+}$ with the multiplicities $g_{1},g_{2},...,g_{r}$. Put
$g:=\sum_{j=1}^{r}g_{j}$. Let $p(\cdot)$ be the polynomial of degree $g$, defined by $p(\lambda):=\prod_{j=1}^{r}(\lambda-\eta_{j})^{g_{j}}$. Put $\mathcal{M}_{+}:=S\mathcal{N}_{+}$, where $\mathcal{N}_{+}$ is the linear manifold of $\mathcal{H}^{2}_{+}$ of all functions $u$ of the form
$u(\lambda):=\frac{p(\lambda)}{(\lambda+i)^{g}}w(\lambda)$, where $w\in\mathcal{H}^{2}_{+}$. Further put 
$\mathcal{T}_{+}:=\mathcal{H}^{2}_{+}\ominus\mathcal{M}_{+}$.

Proof of (i): Let $\zeta$ be an eigenvalue of $B_{+}$. Then a corresponding eigenvector has necessarily the form $f(\lambda)=
\frac{k_{0}}{\lambda-\zeta}$ for som $k_{0}\in\mathcal{K}$ and one has
\[
\int_{-\infty}^{\infty}(\frac{k_{0}}{\lambda-\zeta},S(\lambda)u(\lambda))_\mathcal{K}d\lambda=
\int_{-\infty}^{\infty}\frac{1}{\lambda-\overline{\zeta}}(k_{0},S(\lambda)u(\lambda))_\mathcal{K}d\lambda=0,\quad u\in\mathcal{N_{+}}.
\]
First let $\overline{\zeta}$ be a holomorphic point of $S(\cdot)$. Then one obtains
\[
\int_{-\infty}^{\infty}\frac{1}{\lambda-\overline{\zeta}}(S(\lambda)^{\ast}k_{0},u(\lambda))_\mathcal{K}d\lambda=
2\pi i(S(\overline{\zeta})^{\ast}k_{0},u(\overline{\zeta}))_\mathcal{K}=0,
\]
hence $S(\overline{\zeta})^{\ast}k_{0}=0$ follows, since every $k\in\mathcal{K}$ is possible for $u(\overline{\zeta})$. This means that
$S(\overline{\zeta})^{\ast}$ is not invertible, i.e. $\zeta$ is a pole and (a) is true. Therefore, one can assume that $\overline{\zeta}$ is a pole. Then one gets
\[
\int_{-\infty}^{\infty}(\frac{k_{0}}{\lambda-\zeta},S(\lambda)\frac{p(\lambda)}{(\lambda+i)^{g}}w(\lambda))_\mathcal{K}d\lambda=
\]
\[
\int_{-\infty}^{\infty}\frac{1}{\lambda-\overline{\zeta}}(k_{0},S(\lambda)p(\lambda)\frac{w(\lambda)}{(\lambda+i)^{g}})_\mathcal{K}d\lambda =0=
\]
\[
2\pi i(k_{0},(S(\cdot)p(\cdot))(\overline{\zeta})\frac{w(\overline{\zeta})}{(\overline{\zeta}+i)^{g}})_\mathcal{K}=
2\pi i(k_{0},c(\overline{\zeta})A\frac{w(\overline{\zeta})}{(\overline{\zeta}+i)^{g}})_\mathcal{K},
\]
where $(S(\cdot)p(\cdot))(\overline{\zeta})=c(\overline{\zeta})A,\; c(\overline{\zeta})\neq 0$ and $A\neq 0$, where $A$ is the leading term of the main part of the pole $\overline{\zeta}$, i.e. one obtains $(A^{\ast}k_{0},k)=0$ for all $k\in\mathcal{K}$, hence $A^{\ast}k_{0}=0$, i.e. $A$ is not invertible and (b) is true.

For the reversal let $\zeta$ be a pole of $S(\cdot)$ and $S(\overline{\zeta})^{\ast}k_{0}=0$ or let $\overline{\zeta}$ be a pole and
$A^{\ast}k_{0}=0$ for some $k_{0}\in\mathcal{K}, k_{0}\neq 0$. Then all calculations are reversible.                 

Proof of (ii): Let $\zeta\in\mbox{res}\,B_{+}$. If $\zeta$ is a pole then $\zeta$ is an eigenvalue, hence it cannot be a member of
$\mbox{res}\,B_{+}$. If $S(\zeta)$ exists and $\overline{\zeta}$ is a pole, but $A$ is not invertible then $\zeta$ is again an eigenvalue.
Reversal:\\[3pt]
(a): In this case $S(\zeta)$ and $S(\overline{\zeta})$ exist. Then $\zeta$ is not an eigenvalue, hence 
\[
(B_{+}-\zeta\EINS)^{-1}
\]
exists. According to the ``closed graph theorem'' it is sufficient to show that
\[
\mbox{ima}\,(B_{+}-\zeta\EINS)=\mathcal{T}_{+}
\]
is true, i.e. if $g\in\mathcal{T}_{+}$ then one has to construct a function $f\in\mbox{dom}\,B_{+}$ such that
\[
(B_{+}-\zeta\EINS)f=g.
\]
In any case $f$ is an element from the domain of the generator of the full characteristic semigroup. Therefore it is sufficient to construct $f$ as an element of $\mathcal{T}_{+}=\mathcal{H}^{2}_{+}\ominus\mathcal{M}_{+}$ such that
\[
f(\lambda)=\frac{g(\lambda)-k_{0}}{\lambda-\zeta},
\]
where $k_{0}\in\mathcal{K}$ is a suitable vector. Since $g\bot\mathcal{M}_{+}$, i.e. $g\bot S\mathcal{N}_{+}$ or $S^{\ast}g\bot\mathcal{N}_{+}$ one has
\[
S^{\ast}\frac{\overline{p(\cdot)}g(\cdot)}{(\cdot-i)^{g}}\bot\mathcal{H}^{2}_{+},
\]
i.e. the function
\[
\lambda\rightarrow h(\lambda):= S(\lambda)^{\ast}\frac{\overline{p(\lambda)}}{(\lambda-i)^{g}}g(\lambda)
\]
is an element of $\mathcal{H}^{2}_{-}$. The corresponding expression for $f$ reads
\begin{equation}\label{eq:6}
\frac{S(\lambda)^{\ast}\overline{p(\lambda)}}{(\lambda-i)^{g}}f(\lambda)=(\lambda-\zeta)^{-1}(
h(\lambda)-\frac{S(\lambda)^{\ast}\overline{p(\lambda)}}{(\lambda-i)^{g}}k_{0}).
\end{equation}
Since $\zeta\in\Bbb{C}_{-}$, on has $p(\zeta)\neq 0$. In order that the right hand side of equ.~\eqref{eq:6} is from $\mathcal{H}^{2}_{-}$, one has to put
\[
h(\zeta)=\frac{S(\zeta)^{\ast}\overline{p(\zeta)}}{(\zeta-i)^{g}}k_{0}=
\frac{S(\overline{\zeta})^{-1}\overline{p(\zeta)}}{(\zeta-i)^{g}}k_{0}.
\]
Hence $k_{0}$ is uniquely determined by
\[
k_{0}=\frac{(\zeta-i)^{g}}{\overline{p(\zeta)}}S(\zeta)h(\zeta).
\]
(b): In this case $S(\zeta)$ exists but $\overline{\zeta}$ is a pole of $S(\cdot)$ and $A$ is invertible. The function $h$ is defined as before. Now the equation~\eqref{eq:6} is written in the form
\[
\frac{(S(\cdot)p(\cdot))^{\ast}(\lambda)}{(\lambda-i)^{g}}f(\lambda)=
(\lambda-\zeta)^{-1}(h(\lambda)-\frac{(S(\cdot)p(\cdot))^{\ast}(\lambda)}{(\lambda-i)^{g}}k_{0}).
\]
In this case one has $(S(\cdot)p(\cdot))(\overline{\zeta})=cA$, where $c\neq 0$ because $\overline{\zeta}$ is one of the poles $\eta_{j}$ and $A$ is again the leading term of the main part of the pole $\overline{\zeta}$ i.e.
\[
(S(\cdot)p(\cdot))^{\ast}(\zeta)=\overline{c}A^{\ast}
\]
and one obtains $k_{0}$ uniquely from the equation
\[
h(\zeta)=\frac{\overline{c}}{(\zeta-i)^{g}}A^{\ast}k_{0}
\]
i.e.
\[
k_{0}=(\overline{c})^{-1}(\zeta-i)^{g}(A^{\ast})^{-1}h(\zeta).
\]
\end{proof}


\begin{thebibliography}{8}
\bibitem{ref:1}
Baumg\"artel, H.,
\textit{The Resonance-Decay Problem in Quantum Mechanics}, in:
Geometric Methods in Physics, XXX Workshop 2011, Trends in Mathematics, 165-174, Springer Basel 2013.

\bibitem{ref:2}
Baumg\"artel, H.,
\textit{Resonances of quantum mechanical scattering systems and Lax-Phillips scattering theory},
J. Math. Phys. \textbf{51}, 113508 (2010).

\bibitem{ref:3}
Lax, P.D. and Phillips, R.S.,
\textit{Scattering Theory}
Academic Press, New York 1967.

\bibitem{ref:4}
Bohm, A. and Gadella, M.:
\textit{Dirac Kets, Gamov Vectors and Gelfand Triplets},
Lecture Notes in Physics, Vol. \textbf{348}, Springer Berlin 1989.

\bibitem{ref:5}
Wollenberg, M.,
\textit{On the inverse problem in the abstract theory of scattering},
ZIMM-Preprint Akad. Wiss. DDR, Berlin 1977.

\bibitem{ref:6}
Baumg\"artel, H. and Wollenberg, M.,
\textit{Mathematical Scattering Theory}.
Birkh\"auser Basel Boston Stuttgart 1983.

\bibitem{ref:7}
Baumg\"artel, H.,
\textit{Generalized Eigenvectors for Resonances in the Friedrichs Model and Their Associated Gamov Vectors},
Rev. Math. Phys. \textbf{18}, 61-78 (2006).

\bibitem{ref:8}
Baumg\"artel, H, Kaldass, H. and Komy, S.,
\textit{On spectral properties of the resonances for selected potential scattering systems},
J. Math. Phys. \textbf{50}, 023511 (2009).

\bibitem{ref:9}
Baumg\"artel, H.,
\textit{Spectral and Scattering Theory of Friedrichs Models on the Positive Half Line with Hilbert-Schmidt Perturbations},
Annales Henri Poincare´(AHP)10, Nr.1, 123-143 (2009).
\end{thebibliography}
\end{document}